\documentclass[a4paper,12pt,reqno]{article}
\usepackage[left=25mm,top=25mm,bottom=25mm,right=25mm]{geometry}
\usepackage{dsfont,amscd}
\usepackage{graphicx}

\usepackage[affil-it]{authblk}

\usepackage[all]{xy}
\usepackage{float} 
\usepackage[cyr]{aeguill}
\usepackage[pdfdisplaydoctitle=true,
            colorlinks=true,
            urlcolor=blue,
            citecolor=blue,
            linkcolor=blue,
            pdfstartview=FitH,
            pdfpagemode=None,
            bookmarksnumbered=true]{hyperref}
\usepackage{amssymb}
\usepackage{amsthm}             
\usepackage{amsmath}
\usepackage{todonotes}

\graphicspath{{fig/}}

\newtheorem{thm}{Theorem}[section]
\newtheorem{cor}[thm]{Corollary}
\newtheorem{lem}[thm]{Lemma}

\newtheorem{propy}[thm]{Property}

\newtheorem{con}[thm]{Conjecture}

\newtheorem{defn}[thm]{Definition}

\newcommand{\RR}{\mathbb{R}}                
\newcommand{\HH}{\mathbb{H}}                

\newcommand{\SO}{\mathrm{SO}}               

\newcommand{\V}{\underline}

\newcommand{\bee}{\begin{equation}}
\newcommand{\eee}{\end{equation}}
\newcommand{\ben}{\begin{equation*}}
\newcommand{\een}{\end{equation*}}
\newcommand{\ba}{\begin{eqnarray}}
\newcommand{\ea}{\end{eqnarray}}
\newcommand{\ban}{\begin{eqnarray*}}
\newcommand{\ean}{\end{eqnarray*}}

\hypersetup{
    pdftitle={On the algebraic  structure of isotropic generalized elasticity theories}, 
    pdfauthor={N. Auffray}, 
    pdfsubject={}, 
    pdfkeywords={Group representation, harmonic decomposition, strain-gradient elasticity, stress-gradient theory} 
    pdflang=EN 
    }

\begin{document}

\title{On the algebraic  structure of isotropic generalized elasticity theories}

\author{N. Auffray%
\thanks{Electronic address: \texttt{Nicolas.auffray@univ-mlv.fr}; Corresponding author}}
\affil{Université Paris-Est, Laboratoire Modélisation et Simulation Multi Echelle, MSME UMR 8208 CNRS, 5 bd Descartes, 77454 Marne-la-Vallée, France}

\date{\today}%

\maketitle
\begin{abstract}
In this paper the algebraic  structure of the isotropic $n$th-order gradient elasticity is investigated.  In the classical isotropic elasticity it is well-known that the constitutive relation can be broken down into two uncoupled relations between elementary part of the strain and the stress tensors (deviatoric and spherical). In this paper we demonstrate that this result can not be generalized and since $2$nd-order isotropic elasticity there exist couplings between elementary parts of higher-order strain and stress tensors. Therefore, and in certain way, $n$th-order isotropic elasticity have the same kind of algebraic  structure as anisotropic classical elasticity. This structure is investigated in the case of $2$nd-order isotropic elasticity, and moduli characterizing the behavior are provided.\end{abstract}







\begin{center}
\emph{Dedicated to Prof. Antonio Di Carlo in recognition of his academic activity}
\end{center}

\section{Introduction}

In the last years it has been widely recognized that classical continuum mechanics was unable to describe a variety of important mechanical and physical phenomena. In particular,
the size effects and non-local behaviors due to the discrete nature of matter at a sufficiently small scale, the presence of microstructural defects or the existence of internal constraints cannot be captured by classical continuum mechanics.
Higher gradient models are also needed when continuum models are introduced for describing systems in which strong inhomogeneities of physical properties are present at eventually different length scales (see e.g. \cite{Abu08,Pol07,PS97,TB96,YM10,YCM11,YM12}), in continuum systems in which some \emph{microscopical} degrees of freedom can \emph{capture} a relevant amount of deformation energy (see e.g. \cite{Car05,CAK06}), in fracture (see e.g. \cite{CHO98,Li11}), and so on. In \cite{Min65,FCB11} second strain-gradient elasticity,  was investigated by Mindlin  who shows that a second gradient of strain is needed to describe, in a continuous manner, capillarity and cohesion effects in elastic continuum. Furthermore, it has recently been proved, that some microstructure can be specifically designed to render any higher-grade effects predominant \cite{Asi03}. Aside from strain-gradient theories, gradient effects have recently been considered both in the context of plate theory \cite{LS11a,LS11b}, and in the context of stress-gradient \cite{FK12}. 

In the present contribution $n$th-order strain-gradient isotropic ($\SO(3)$-invariance) elasticity is considered. By $n$th-grade elastic materials we mean those materials whose mechanical response depends on the present value of the first $n$ deformation-gradients \cite{ISM12,PGV13}: for $n = 1$ classical elasticity tensor is retrieved, while for case $n = 2$ Mindlin's strain-gradient elasticity \cite{Min64,ME68,ISV09,MB11a,ALH13} is obtained. At the present time, few results (see e.g.\cite{ISM12,OA13a,OA13b}) concerning a general theory of $n$th-gradient elastic material are known. Therefore a general picture of this extended theory has to be drawn. The aim of the present paper is to provide a general result concerning the algebraic  structure of isotropic $n$th-order elasticity tensors. More specifically we will show that, contrary to classical elasticity, for $n$th-order isotropic elasticity the harmonic parts\footnote{As detailed latter on, the harmonic decomposition is the correct generalization of the deviatoric and spheric decomposition of a second-order symmetric tensor.} of higher-order strain and stress tensors are always coupled. In classical elasticity these kinds of coupling  appear for anisotropic classes greater or equal to transverse isotropy \cite{Wal84}. Therefore, as already observed in the case of strain-gradient elasticity \cite{MB11a}, the algebraic  structure of $n$th-order elasticity tensors is very similar to  anisotropic classical elasticity tensors.\\

\noindent\textbf{The paper is organized as follows} In a first section the general context of our study will be presented. After introducing the framework of $n$th-gradient elastic materials, the harmonic decomposition, which is a central tool in our study, is presented together with a method to easily compute its structure. In a second time, the Walpole and Kelvin tensor decompositions are defined and their links investigated. \textbf{This point is important since these two decompositions seem to be often merged}. The section is concluded by stating (without proofs at this stage) our main results. In the next section (\S.\ref{sec:Ex}) two specific situations are studied: the classical and the second-order elasticity. For the second-order elasticity, and following \cite{Wal84,MB11a}, an intrinsic representation, based on a kinematic interpretation of strain-gradient tensor is provided. This interpretation give rises to $5$ moduli associated with \emph{elementary} mechanisms. The last section  is   devoted to the proofs of our results. In appendices some details and explicit constructions concerning  second-order elasticity tensors are provided.\\

\noindent\textbf{Notations} Let us  define some notations that will be used throughout the paper, other less important one will be introduced in the core of the text. 
The following matrix groups will be considered in the present paper :
\begin{itemize}
\item $\mathrm{GL}(3)$: the group of all the invertible transformations of $\mathbb{R}^{3}$, i.e. if $\mathrm{F}\in\mathrm{GL}(3)$ then $\mathrm{det}(\mathrm{F})\neq0$;
\item $\mathrm{O}(3)$ : the orthogonal group, i.e. the group of all isometries of $\mathbb{R}^{3}$ i.e. if $\mathrm{Q}\in\mathrm{O}(3)$ $\mathrm{det}(\mathrm{Q})\pm 1$ and $\mathrm{Q}^{-1}=\mathrm{Q}^{T}$, where the superscript $^{T}$ denotes the transposition.
\item $\mathrm{SO}(3)$ : the special orthogonal group, i.e. the subgroup of $\mathrm{O}(3)$ of elements satisfying $\mathrm{det}(\mathrm{Q})= 1$. This is the group of 3D rotations;
\item $\mathcal{O}$  : the octahedral group, i.e. the group of rotations that let a cube invariant.
\end{itemize} 
Vector spaces will be denoted using blackboard fonts, and their tensorial order indicated by using formal indices. Generic tensor spaces will be denoted $\mathbb{T}$, and space of harmonic tensors $k$-th order harmonic tensor $\mathbb{H}^{k}$. Harmonic tensor spaces are $\mathrm{SO}(3)$-invariant, their elements are \emph{completely symmetric} and \emph{traceless} $k$-th order tensors \cite{Bac70,JCB78,FV96}. In 3D we have $\mathrm{dim}(\mathbb{H}^{k})=2k+1$,\textbf{ algorithms to compute the harmonic decomposition can be found in the following references}\cite{Spe70,Jar03,Auf08b,Auf13,Ste94}. The precise meaning of these spaces will be detailed in the text.
When needed, index symmetries of both space and their element are expressed as follows:  $(..)$  indicates invariance under permutation of the indices in parentheses, and $\underline{..}$ indicates invariance with respect to permutations of the underlined blocks.\\

In the following, the  tensorial product of $\RR^{3}$ is denoted $\otimes$, while $\otimes^{k}$ indicates the $k$-order tensor product and  $S^k$ its completely symmetrized version. Classically the spatial derivation will be denoted by $\nabla$.

\section{Setting of the problem and main results}\label{sec:set}

\subsection{$n$th-gradient elastic materials}

\subsubsection*{Primary and dual quantities}
In $n$th grade elasticity, the mechanical response is supposed to depend on the first $n$ deformation-gradients. Those quantities are defined recursively 
\ben
\begin{cases}
\varepsilon^{(1)}=\varepsilon\\
\varepsilon^{(k+1)}=\varepsilon^{(k)}\otimes\nabla\\
\end{cases}
\een
where $\varepsilon$ and $\varepsilon^{(k)}$ denote, respectively, the classical symmetric strain tensor and the $k$th-order strain tensor.  Due to the fact that $\varepsilon$ is a second-order symmetric tensor, $\varepsilon^{(k)}$ is a $(k+1)$th-order tensor that belongs to the space $\mathbb{T}^{(k)}$ defined as follows:
\bee\label{eq:def}
\mathbb{T}^{(k)}=S^2(\RR^{3})\times S^{k-1}(\RR^{3})\subseteq\otimes^{k+1}\RR^{3}
\eee
Elements of $\mathbb{T}^{(k)}$ are $(k+1)$th-order tensors symmetric with respect to their 2 first indices and their $k-1$ last. 
Using the convention for index symmetries defined in the introduction, elements of $\mathbb{T}^{(k)}$ have the following shape:
\ben
\mathrm{T}_{(i_{1}i_{2})(i_{3}\ldots i_{k-1})}
\een
By duality  $\sigma^{(k)}$ is defined as the energetic conjugate of $\varepsilon^{(k)}$, and for $k=1$ the symmetric Cauchy stress tensor is retrieved. It is worth noting that the $\sigma^{(k)}$ is not the gradient of $\sigma^{(k-1)}$.\\ 

\subsubsection*{Constitutive law}
The object of  our study concerns the following $n$th-order elasticity relation:
\ben
\sigma^{(n)}=\mathrm{C}^{(n)}\varepsilon^{(n)}
\een
where  $\mathrm{C}^{(n)}$ is a $2(n+1)$th-order tensor. It is worth noting that this tensorial law is just a part of the complete $n$th-order gradient elasticity law which also contains coupling terms with the other deformation gradients \cite{Min64,Min65}. Therefore, in this paper, \emph{$n$th-order elasticity relation}  only refers to this part of the complete $n$th-order gradient elasticity law. For $n=1$ the classical Hooke's law is obtained, and for $n=2$ we obtained the second-order elasticity studied, for example, in \cite{ISV09,ALH13}. The vector space of $\mathrm{C}^{(n)}$ is defined as the space of symmetric endomorphism of $\mathbb{T}^{(n)}$. In terms of tensorial products, this space is constructed as
\ben
\mathbb{C}^{(n)}=S^2(\mathbb{T}^{(n)})\subseteq\otimes^{2(n+1)}\RR^{3}
\een
Hence elements of $\mathbb{C}^{(n)}$ inherit the minor symmetries of $\sigma^{(n)}$ and $\varepsilon^{(n)}$, completed by the major one. In other terms, elements of $\mathbb{C}^{(n)}$ have  the following shape:
\ben
\mathrm{C}_{\underline{(i_{1}i_{2})(j_{1}\ldots j_{n-1})}\ \underline{(k_{1}k_{2})(l_{1}\ldots l_{n-1})}}
\een
For our study, we need to break tensor spaces into elementary building blocks. The definition of these elementary spaces depends on the considered group action.
In the present situation, tensor spaces will be decomposed into $\SO(3)$-invariant spaces. This decomposition is the (correct) generalization of the decomposition of second-order symmetric tensors into deviatoric and spheric parts.

\subsection{The harmonic decomposition}
The harmonic decomposition has been widely used in the mechanical community for studying anisotropic elasticity \cite{Bac70,BKO94,FV96,FV97}. In order to present our results in self-contained way, basic definitions concerning this decomposition are summed up here. A more general and rigorous presentation can be found in \cite{JCB78,Ste94}\textbf{, and some historical considerations concerning its uses in mechanics in} \cite{FV96}.
In $\RR^3$, under $\SO(3)$-action any tensor space $\mathbb{V}$ can be decomposed orthogonally into a sum of harmonic tensor spaces of different orders\footnote{In the harmonic decomposition of a tensor space, the equality sign means that there exists an isomorphism $h$ between the right- and the left-hand side of the decomposition. In order to avoid the use of too many notations we do not use a specific sign to indicate this isomorphism.}:
\[
\mathbb{V}=\bigoplus_{i=0}^{p}\alpha_{i}\mathbb{H}^{i} 
\]
where $p$ indicates the tensorial order of $\mathbb{V}$, $\mathbb{H}^{i}$ is the vector space of $i$th-order harmonic tensors and $\alpha_{i}$ indicates the number of copies of  $\mathbb{H}^{i}$ in the decomposition.
Elements of $\mathbb{H}^{i}$ are $i$th-order \emph{completely symmetric} \emph{traceless} tensors, the dimension of their vector space is $2i+1$.
The denomination harmonic is related to a classical isomorphism in $\RR^3$ between the space of harmonic polynomials (i.e. polynomials with null Laplacian) of degree $i$ and the space of $i$th-order completely symmetric traceless tensors. If needed, an explicit isomorphism, denoted $h$ in the following, can be computed using, for example, an algorithm proposed by Spencer \cite{Spe70,Jar03}. The following  property is important \cite{GSS88}
\begin{propy}\label{prop_unic}
The isomorphism $h$ that realizes the harmonic decomposition is uniquely defined iff $\alpha_{i}\leq1$ for $i\in[0,p]$.
\end{propy}
According to the problem under investigation, the explicit knowledge of $h$ may be required or not. This point is important because the explicit computation of an isomorphism is a task which complexity increases very quickly with the tensorial order. At the opposite, the determination of the structure of the decomposition, i.e. the number of spaces together with their multiplicities is almost straightforward using the Clebsch-Gordan product.\\

\subsection{Clebsch-Gordan product}

For our need, in order to determine the algebraic  structure of $n$th-order elasticity tensors, we only need to know the structure of the harmonic decomposition.
To compute it we use the tensorial product of group representations. More details can be found in \cite{JCB78,Auf08b}. The computation rule is simple. Consider two harmonic tensor spaces $\mathbb{H}^{i}$ and $\mathbb{H}^{j}$, whose product space is noted $\mathbb{T}^{i+k}:=\mathbb{H}^{i}\otimes\mathbb{H}^{j}$. This space admits the following $\mathrm{SO}(3)$-invariant decomposition:
\bee\label{eq:CG_dec}
\mathbb{G}^{i+j}=\mathbb{H}^{i}\otimes\mathbb{H}^{j}=\bigoplus_{k=|i-j|}^{i+j}\mathbb{H}^{k}
\eee
In the following this rule will be referred to as the Clebsh-Gordan product.
For example, consider $\mathbb{H}_{a}^{1}$ and $\mathbb{H}_{b}^{1}$ two different first-order harmonic spaces. Elements of such spaces are vectors. According to formula \eqref{eq:CG_dec}, the $\mathrm{SO}(3)$-invariant decomposition of $\mathbb{T}^{2}$ is:
\ben
\mathbb{G}^{2}=\mathbb{H}_{a}^{1}\otimes\mathbb{H}_{b}^{1}=\mathbb{H}^{2}\oplus\mathbb{H}^{1}\oplus\mathbb{H}^{0}
\een
As an example, the tensorial product of two vector spaces generates a second-order tensor space. The resulting structure is composed of a scalar ($\mathbb{H}^{0}$), a vector ($\mathbb{H}^{1}$) and a deviator ($\mathbb{H}^{2}$).
This computation rule has to be completed by the following property \cite{JCB78}
\begin{propy}\label{prop1}
The decomposition of an even-order (resp. odd-order) completely symmetric tensor, i.e. invariant under any index permutation, only contains even-order (resp. odd-order) harmonic spaces.
\end{propy}
Applying these rules, the harmonic structure of $\mathbb{T}^{(n)}$ (the space of $n$th-order stress and strain tensors), is easily determined. Results for low order spaces are provided in the following table:
\ben\label{eq:Tab_Dec}
   \begin{array}{|c||c|c|c|c|}
       \hline
       \text{Vector space}&\text{Elements}&\text{Harmonic decomposition}&\text{Dimension}&\text{Unicity}\\\hline
       \mathbb{T}^{1}&\varepsilon_{(ij)}&\HH^2\oplus\HH^0&6&\text{Y}\\ \hline
       \mathbb{T}^{2}&\varepsilon_{(ij),k}&\HH^3\oplus\HH^2\oplus2\HH^1&18&\text{N}\\ \hline
       \mathbb{T}^{3}&\varepsilon_{(ij),(kl)}&\HH^4\oplus\HH^3\oplus3\HH^2\oplus\HH^1\oplus2\HH^0&36&\text{N}\\ \hline
   \end{array}
\een
It can be observed, on these examples, that, according to the property \ref{prop_unic}, $\mathbb{T}^{1}$ is the only space having an uniquely defined harmonic decomposition. This fact will be important for the forthcoming discussion.

\subsection{Tensor decompositions}
In the following\footnote{In this subsection all $(n)$ superscripts will be dropped down to enlighten notations.} two different tensor decompositions will be discussed, the \emph{Spectral or Kelvin's decomposition} and the \emph{Harmonic-induced or Walpole's decomposition}. To that aim, some orthonormal bases have to be introduced. On $\mathbb{T}$ three different bases will be considered:
\begin{itemize}
\item $\mathcal{S}$: the spatial-basis, constructed by tensor products of the canonical one (see Appendix \ref{Sec:MatRep} for an explicit example);
\item $\mathcal{E}(\mathrm{C})$: the eigen-basis, constructed from the eigenvectors of an element $\mathrm{C}\in\mathbb{T}\otimes^S\mathbb{T}$, the notation keeps track of that dependence; 
\item $\mathcal{H}(\mathbb{T})$: the harmonic-basis, constructed from an $\SO(3)$-irreducible spaces decomposition of $\mathbb{T}$, the notation keeps track of that dependence. 
\end{itemize}
Similarly on $\mathbb{T}\otimes^S\mathbb{T}$ the following bases will be considered:
\begin{itemize}
\item $\mathcal{S}\otimes^S\mathcal{S}$: in this basis the components of $\mathrm{C}$ are denoted $C_{\mathcal{S}}$. This corresponds to the classical matrix representation of a constitutive tensor \cite{MC90,ALH13};
\item $\mathcal{E}(\mathrm{C})\otimes^S\mathcal{E}(\mathrm{C})$: in this basis the components of $\mathrm{C}$ are denoted $C_{\mathcal{E}}$. In this basis the matrix of $\mathrm{C}$ is diagonal. This representation correspond to the \emph{Spectral or Kelvin's decomposition}, which has been widely used in the context of classical elasticity (see e.g. \cite{Ryc84,MC90,BBS07});
\item $\mathcal{H}(\mathbb{T})\otimes^S\mathcal{H}(\mathbb{T})$: in this basis the components of $\mathrm{C}$ are denoted $C_{\mathcal{H}}$. The associated matrix representation has been less investigated in the literature, but it can be associated with some constructions of Walpole \cite{Wal84,MB11a,MB13}. Hence this decomposition will be referred to as the \emph{Harmonic-induced or Walpole's decomposition}\footnote{To be more exact, this decomposition corresponds to the isotropic Walpole decomposition, since Walpole decomposition can be associated with any $G$-invariant decomposition of $\mathbb{T}$, in which $G$ is a subgroup of $\mathrm{SO}(3)$.}.
\end{itemize}

The links that relate the different bases are indicated in the following diagram:
\ben
\forall \mathrm{C}\in\mathbb{T}\otimes^S\mathbb{T},\quad
  \xymatrix{
    \mathcal{S}\otimes^S\mathcal{S} \ar[r]^{d}\ar[rd]^{\phi(h)}& \mathcal{E}(\mathrm{C})\otimes^S\mathcal{E}(\mathrm{C})\ar@{~>}[r]&\mathcal{E}(\mathrm{C})&\\
                          & \mathcal{H}(\mathbb{T})\otimes^S\mathcal{H}(\mathbb{T})\ar[u]_{d^{\star}}&\mathcal{H}(\mathbb{T})\ar[u]^{\psi(d^{\star})}\ar@{~>}[l]&\ar[l]^{h}\ar[lu]_{\psi(d)} \mathcal{S}
    }
\een
where $d$ and $d^{\star}$ indicate diagonalization processes, $\psi(d)$ and $\psi(d^{\star})$ the related changes of basis, and $\phi(h)$ the change of basis induced by the harmonic decomposition $h$ of $\mathbb{T}$. In this diagram the wavy arrows indicate that the change of basis in the destination space is induced by an operation in the original space. 
As a consequence the different components of $\mathrm{C}$ are related in the following way:
\ben
  \xymatrix{
    &C_{\mathcal{E}}\\
    C_{\mathcal{S}} \ar[r]^{\phi(h)}\ar[ru]^{d}   &C_{\mathcal{H}}\ar[u]_{d^{\star}}\\
    }
\een
It has to be observed that the former diagram commutes, that is:
\ben
d=d^{\star}\circ\phi(h)
\een
Therefore $C_{\mathcal{H}}$ maybe considered as an intermediary step in the transformation from $C_{\mathcal{S}}$ to $C_{\mathcal{E}}$.\\

\emph{The basis $\mathcal{H}(\mathbb{T})$ is determined by the $\SO(3)$-structure of $\mathbb{T}$, while $\mathcal{E}(\mathrm{C})$ is specific to each element of $\mathbb{T}\otimes^S\mathbb{T}$. As a consequence the isotropic Walpole's decomposition is determined by the $\SO(3)$-structure\footnote{Or more generally a $G$-invariant decomposition for anisotropic tensors, in which $G$ is a subgroup of $\mathrm{SO(3)}$.} of $\mathbb{T}$ meanwhile the Kelvin's decomposition is specific to each element of $\mathbb{T}\otimes^S\mathbb{T}$.}\\ 

Therefore there is an intrinsic dissymmetry between these two decompositions. But as in some specific situations they coincide, these decompositions are often merged and the Walpole's one is sometimes referred to as spectral, which is false. For classical elasticity, these decompositions are identical only for both the isotropic and the cubic systems. Furthermore, and as it will be demonstrated, the Walpole representation for higher-order isotropic elasticity is no more diagonal.\\

\subsection{Main results}

In \S.\ref{sec:proof} the following theorem is proved:
\begin{thm}[Theorem 1]\label{th1}
If the harmonic decomposition of $\mathbb{T}$ is unique, and if $\mathrm{C}$ is isotropic then its matrix representation $C_{\mathcal{H}}$ is diagonal. In such a case $d^{\star}$ reduces to a permutation and the Walpole and the Kelvin decompositions are equivalent.
\end{thm}
This theorem can be reformulated in another way. To that aim let $\left[\mathbb{T}\otimes^S\mathbb{T}\right]^{\SO(3)}$ denotes the subspace of isotropic elements of $\mathbb{T}\otimes^S\mathbb{T}$, and $\sharp^{\SO(3)}(\mathbb{T})$ the number of harmonic components in the decomposition of $\mathbb{T}$. We have the alternative formulation to Theorem \ref{th1}:
\begin{thm}[Theorem 1']\label{th2}
Let consider $\mathrm{C}\in\left[\mathbb{T}\otimes^S\mathbb{T}\right]^{\SO(3)}$, the Walpole and the Kelvin representation of $\mathrm{C}$ are equivalent iff
\ben
\dim(\left[\mathbb{T}\otimes^S\mathbb{T}\right]^{\SO(3)})=\sharp^{\SO(3)}(\mathbb{T})
\een
\end{thm}
Even if we did not investigate in that direction in the present paper, we may conjecture that the Theorem 1' can be extended to anisotropic situations, i.e.
\begin{con}\label{con}
Let consider $\mathrm{C}\in\left[\mathbb{T}\otimes^S\mathbb{T}\right]^{G}$, the space of $G$-invariant symmetric endomorphisms of $\mathbb{T}$ and $\sharp^{G}(\mathbb{T})$ be the number of $G$-irreducible components in the decomposition of $\mathbb{T}$. The Walpole and the Kelvin representation of $\mathrm{C}$ are equivalent iff
\ben
\dim(\left[\mathbb{T}\otimes^S\mathbb{T}\right]^{G})=\sharp^{G}(\mathbb{T})
\een
\end{con}
For example this situation occurs for cubic classical elasticity in which the $3$ different coefficients correspond to the number of $\mathcal{O}$-irreducible spaces of $\mathbb{T}^{1}$. As well-known \cite{Wal84} in this situation the Walpole and the Kelvin's decomposition coincide.\\

Now we have to determine in which situations the condition of the former theorems is verified. The answer is given by the following theorem
\begin{thm}[Theorem 2]\label{th3}
Let $\mathbb{T}$ be a $n$th-order tensor space only endowed with minor symmetries, then the harmonic decomposition of $\mathbb{T}$ is unique or, equivalently:
\ben
\dim(\left[\mathbb{T}\otimes^S\mathbb{T}\right]^{\SO(3)})=\sharp^{\SO(3)}(\mathbb{T})
\een
if $\mathbb{T}=S^n(\RR^{3}),\  n\geq0$ or $\mathbb{T}=\otimes^2(\RR^{3})$.
\end{thm}
A direct application of Theorem \ref{th2} shows that only $\mathbb{T}^{1}$ meets this requirement. Therefore for $n$th-gradient elasticity,  Theorem \ref{th1} is true only for first order elasticity, and since $n=2$ there exist couplings between harmonic components of higher order stresses and strains. In classical elasticity this kind of coupling appears for transverse isotropy and lower symmetry classes (except the cubic one). In other terms, $n$th-order isotropic elasticity have the same kind of algebraic  structure as anisotropic classical elasticity.

\section{Illustrations : First and second-order elasticity}\label{sec:Ex}

In this section some practical illustrations of our results will be given. For second-order elasticity this explicit construction gives insights into the physic encoded by the second-order constitute law. As a result, we obtain the explicit expressions of five moduli, four of them related to \emph{elementary physical mechanisms} while the last one is a coupling modulus between these mechanisms. In order not to lengthen too much the section, some explicit constructions for second-order elasticity, such as the $18\times 18$ passage matrix, are postponed to the appendix.  

\subsection{$n=1$: Classical elasticity}
In classical elasticity, the constitutive relation reads\footnote{In this subsection, to avoid cumbersome notations, superscripts that indicate the order of relations are dropped down.}
\bee\label{eq:Ela_Cla}
\sigma=\mathrm{C}:\varepsilon,\quad \quad \sigma,\varepsilon\in\mathbb{T}^{1}=\mathbb{T}_{(ij)},\ \mathrm{C}\in S^2(\mathbb{T}^{(1)})=\mathbb{T}_{\underline{(ij)}\ \underline{(kl)}}
\eee
As well-known the isotropic elasticity tensor is defined by 2 coefficients, and $\mathbb{T}_{(ij)}$ contains two harmonic spaces. Hence Theorem 1' applies and the Walpole and the Kelvin's decomposition coincide.
It is known that, for isotropic behavior, $\mathrm{C}$ can be expressed in the following operational form:
\bee\label{eq:Ela_KG}
\mathrm{C}=3K\mathrm{P}^{\mathbb{H}^0}+2G\mathrm{P}^{\mathbb{H}^2}
\eee
where $K$ and $G$ are respectively known as the bulk and the shear modulus, while $\mathrm{P}^{\mathbb{H}^0}$ and $\mathrm{P}^{\mathbb{H}^2}$ are the hydrostatic and the deviatoric projectors. Therefore $\mathrm{P}^{\mathbb{H}^0}$ and $\mathrm{P}^{\mathbb{H}^2}$ are intrinsic basis vectors for isotropic classical elasticity tensors. These basis vectors satisfy the following properties:\\
\noindent$\bullet$ Partition of the identity\footnote{In this expression $\mathrm{Id}_{\mathbb{T}^{1}}$ denotes the identity operator for the space $\mathbb{T}^{1}$.}:
\ben
\mathrm{Id}_{\mathbb{T}^{1}}=\mathrm{P}^{\mathbb{H}^2}+\mathrm{P}^{\mathbb{H}^0}
\een
\noindent$\bullet$ Multiplication table:
\ben
   \begin{array}{|c||c|c|}
       \hline
       \times&\mathrm{P}^{\mathbb{H}^2}&\mathrm{P}^{\mathbb{H}^0}\\ \hline \hline
       \mathrm{P}^{\mathbb{H}^2}&\mathrm{P}^{\mathbb{H}^2}&0\\\hline
       \mathrm{P}^{\mathbb{H}^0}&0&\mathrm{P}^{\mathbb{H}^0}\\\hline
   \end{array}
\een
According to these properties the constitutive law \eqref{eq:Ela_Cla} can be broken down into two uncoupled relations
\bee\label{eq:Ela-Dec}
\begin{cases}
\sigma^{\mathbb{H}^2}=2G\varepsilon^{\mathbb{H}^2}\\
\sigma^{\mathbb{H}^0}=3K\varepsilon^{\mathbb{H}^0}
\end{cases}
\eee
This decomposition of classical elasticity is well-known and has been established by numerous methods in the literature (see e.g. \cite{Ryc84,Wal84,MC90,BBS07}).
If we consider shear and spherical strains (resp. stresses) as elementary strains (resp. stresses), this decomposition means that for isotropic system these mechanisms are not coupled. From a computational point of view, this representation has many advantages since the inversion of the constitutive matrix is direct \cite{Wal84,MB11a}.  

The following observation is important: \emph{the system \eqref{eq:Ela-Dec} means that the constitutive equation \eqref{eq:Ela_Cla} expressed in the basis induced by the harmonic decomposition of $\mathbb{T}^{1}$ is diagonal.} Therefore in this case the Walpole and Kelvin decompositions coincide. This fact was indeed clear from the Theorem 1'.

\subsection{$n=2$: Second-order elasticity}

It has been observed that for classical isotropic elasticity elementary parts of stress and strain are not coupled. Let consider now the linear relation of second order elasticity:
\ben\label{eq:Ela_Sec}
\tau=\mathrm{A}\therefore\eta,\ \quad \tau,\eta\in\mathbb{T}^{2}=\mathbb{T}_{(ij)k},\ \mathrm{A}\in S^2(\mathbb{T}^{(2)})=\mathbb{T}_{\underline{(ij)k}\ \underline{(lm)n}}
\een
where $\tau,\mathrm{A},\eta$ stand respectively for $\sigma^{(2)},\mathrm{C}^{(2)},\varepsilon^{(2)}$. Therefore both $\tau$ and $\eta$ are third-order tensors, while $\mathrm{A}$ is a sixth-order one.\\

\subsubsection*{$\mathbb{T}^{2}$ harmonic decomposition}\label{ss:DecT2}
As we are interested in the Walpole decomposition of $\mathrm{A}$, the first point is to analyze the harmonic decomposition of both $\eta$ and $\tau$.
As calculated in the table \eqref{eq:Tab_Dec}, $\mathbb{T}^{2}$ admits the following harmonic decomposition: 
\bee\label{eq:HarDec}
\mathbb{T}^{2}\cong\mathbb{H}^{3}\oplus\mathbb{H}^{2}\oplus\mathbb{H}^{1,a}\oplus\mathbb{H}^{1,b}
\eee
It can be observed that any element of this space contains $2$ vectorial parts, therefore the isomorphism associated with the decomposition is not unique. From Theorem 1 we know that the Walpole decomposition will not be diagonal and, contrary to classical elasticity, isotropic elementary couplings between harmonic spaces will occur.
This coupling only concerns the vector parts since both $\HH^{3}$ and $\HH^{2}$ are uniquely defined.
Therefore many constructions are possible, but among them some are more natural since they give a physical interpretation to the harmonic decomposition. In the following we will only consider a kinematic interpretation of the decomposition (see \cite{Auf13} for another interpretations). The approach consists in splitting $T_{(ij)k}$ first into a complete symmetric part and a remainder one before proceeding to the harmonic decomposition. This approach is summed-up by the following diagram:
\ben
\xymatrix{
       &  \mathbb{T}_{(ij)k} \ar[ld]_{\mathrm{Sym}}\ar[rd]^{\mathrm{Asym}=\mathrm{Id}-\mathrm{Sym}} &  \\
  \mathbb{S}_{(ijk)} \ar[d]^{\mathcal{H}}  &    & \mathbb{R}_{ij} \ar[d]^{\mathcal{H}} \\
  \mathbb{H}^{3}\oplus\mathbb{H}^{1}_{a}&&\mathbb{H}^{\sharp 2}\oplus\mathbb{H}^{1}_{b}}
\een
where $\mathrm{Sym}$, $\mathrm{Asym}$ and $\mathcal{H}$ respectively stand for the symmetrization, anti-symmetrization and the harmonic decomposition processes.
$T_{(ij)k}$ is first split into a full symmetric tensor and an asymmetric one:
\ben
T_{(ij)k}=S_{ijk}+\frac{1}{3}\left(\epsilon_{jkl}R_{li}+\epsilon_{ikl}R_{lj}\right)
\een 
where $\epsilon_{ijk}$ denotes the Levi-Civita symbol in 3D.
The space of full symmetric third-order tensors is $10$-dimensional meanwhile the space of the remaining one is $8$ dimensional, those spaces are in direct sum.
In the strain-gradient literature \cite{Min64} the complete symmetric part $S_{(ijk)}$, defined:
\ben
S_{(ijk)}=\frac{1}{3}(T_{(ij)k}+T_{(ki)j}+T_{(jk)i})
\een
is related to the \emph{stretch-gradient} part of $\mathrm{T}_{(ij)k}$. 
The remaining traceless non-symmetric part $R_{ij}$:
\ben
R_{ij}=\epsilon_{ipq}T_{(jp)q}
\een
is the \emph{the rotation-gradient} part of $\mathrm{T}_{(ij)k}$. At this stage, this decomposition coincides with the type III formulation of Mindlin strain gradient elasticity \cite{ME68}. According to its words, this  \emph{third form of the theory is the most convenient one for reduction to the theory in which the potential energy-density is a function of the strain and the gradient of the rotation}.  The terms $\mathrm{V}^{\nabla \mathrm{rot}}$ is the components of the strain gradient tensor that give rise to couple-stress \cite{Min64}.

This first decomposition, which is sometimes referred to as the Schur decomposition, is $\mathrm{GL}(3)$-invariant, meaning that each component is $\mathrm{GL}(3)$-irreducible. In other terms, this decomposition of the strain-gradient tensors into two \emph{mechanisms} (stretch-gradient and rotation-gradient) is preserved under any invertible transformation.
Under $\mathrm{SO}(3)$-action each part can further be decomposed into harmonic components by removing their different traces:
\begin{itemize}
\item $\mathbb{S}_{(ijk)}$ splits into a $3$rd-order deviator ($\dim  \mathbb{H}^{3}=7$) and a vector ($\dim\mathbb{H}_{a}^{1}=3$);
\item $\mathbb{R}_{ij}$ splits into a (pseudo-)deviator ($\dim  \mathbb{H}^{2}=5$) and a vector ($\dim\mathbb{H}_{b}^{1}=3$).
\end{itemize}

\noindent\textbf{\emph{Stretch-gradient tensors:}}\\
The space $\mathbb{S}_{(ijk)}$ is isomorphic to $\mathbb{H}^{3}\oplus\mathbb{H}^{1}_{\nabla str}$, hence the isomorphism that realizes the decomposition is unique. Doing some algebra we obtain
\ben
S_{(ijk)}=H_{(ijk)}+\frac{1}{5}\left(V^{\nabla str}_{i}\delta_{(jk)}+V^{\nabla str}_{j}\delta_{(ik)}+V^{\nabla str}_{k}\delta_{(ij)}\right)
\een
with 
\ben
V^{\nabla str}_{i}=S_{(pp)i}=\frac{1}{3}\left(T_{ppi}+2T_{ipp}\right)\ ;\ 
H_{(ijk)}=S_{(ijk)}-\frac{1}{5}\left(V^{\nabla str}_{i}\delta_{(jk)}+V^{\nabla str}_{j}\delta_{(ik)}+V^{\nabla str}_{k}\delta_{(ij)}\right)
\een
In this interpretation $\mathrm{V}^{\nabla str}$ is the vector part of \emph{the stretch-gradient tensor}.\\

\noindent\textbf{\emph{Rotation-gradient tensors:}}\\
The space $\mathbb{R}_{ij}$ is isomorphic to $\mathbb{H}^{2}\oplus\mathbb{H}^{1}_{\nabla rot}$, hence the isomorphism that realizes the decomposition is unique. Doing some algebra we obtain
\ben
R_{ij}=H_{(ij)}+\epsilon_{ijp}V^{\nabla rot}_{p}
\een
with 
\ben
V^{\nabla rot}_{i}=\frac{1}{2}\epsilon_{ipq}(R_{pq}-R_{qp})=\frac{1}{2}\left(T_{ppi}-T_{ipp}\right)\ ;\  H_{(ij)}=R_{ij}-\frac{1}{2}\epsilon_{ijp}V^{\nabla rot}_{p}=\frac{1}{2}(R_{pq}+R_{qp})
\een
In this formulation $\mathrm{V}^{\nabla rot}$ is the vector part of \emph{the rotation-gradient tensor}, and is embedded in the third-order tensor in the following way:
\ben
T(\V{V}^{\nabla rot})_{ijk}=\frac{1}{3}\left(-V^{\nabla rot}_{i}\delta_{(jk)}-V^{\nabla rot}_{j}\delta_{(ik)}+2V^{\nabla str}_{k}\delta_{(ij)}\right)
\een
Therefore a physical meaning can be given to the harmonic decomposition of $\mathbb{T}^{(2)}$, since the strain-gradient tensor encodes two orthogonal effects: stretch-gradient and rotation-gradient. These effects are canonically defined and preserved under invertible changes of variables. The harmonic decomposition of these elementary effects correspond to their decomposition in spherical harmonics \cite{JCB78}.\\

\subsubsection*{Walpole decomposition of second-order elasticity tensor}
Using the expressions of $H_{(ijk)}$, $V^{\nabla str}_{i}$, $R_{ij}$, $V^{\nabla rot}_{i}$ one can easily build a transformation matrix $P$ from  $\mathcal{S}(\mathbb{T}^2)$, the spatial basis of $\mathbb{T}^2$, to $\mathcal{H}(\mathbb{T}^2)$, a basis related to its harmonic decomposition  (see the Appendix for the explicit expressions of the chosen basis vectors). As the decomposition is not unique the result depends on the physical interpretation chosen for the vector parts. In this new basis, $\mathrm{A}$ has the following expression:
\ben
A_{\mathcal{H}(\mathbb{T}^2)}=
\begin{pmatrix}
m_{s_{3}}\mathbf{\mathrm{Id}}_{7}&\mathrm{0}&\mathrm{0}&\mathrm{0}\\
\mathrm{0}&m_{s_{1}}\mathbf{\mathrm{Id}}_{3}&\mathrm{0}&m_{c_{1}}\mathbf{\mathrm{Id}}_{3}\\
\mathrm{0}&\mathrm{0}&m_{r_{2}}\mathbf{\mathrm{Id}}_{5}&\mathrm{0}\\
\mathrm{0}&m_{c_{1}}\mathbf{\mathrm{Id}}_{3}&\mathrm{0}&m_{r_{1}}\mathbf{\mathrm{Id}}_{3}
\end{pmatrix}_{\mathcal{H}(\mathbb{T}^2)}
\een
where $\mathbf{\mathrm{Id}}_{k}$ indicates the $k$th-order identity matrix. In this particular basis, it clearly appears that the isotropic behavior is coupled. Hence contrary to classical elasticity (both in 2D and 3D), in the isotropic system, the matrix of $\mathrm{A}$ is not diagonal. Therefore, in the generic situation, a coupling always exists between $\mathbb{H}^{1}_{\nabla \mathrm{str}}$ and $\mathbb{H}^{1}_{\nabla \mathrm{rot}}$. A physical consequence is that, for example, even for isotropic material, pure stretch-gradient generates couple-stress \cite{Min64,Ger73}. This coupling will obviously disappear if the coupling modulus $m_{c_{1}}$ is equal to 0.\\

Using this matrix representation, and  as usually done in classical elasticity (c.f. relation \eqref{eq:Ela_KG}), $\mathrm{A}$ can be rewritten in an operational intrinsic form. To that aim let us consider $\{h_{i}\}$ the orthonormal basis of $\mathcal{H}(\mathbb{T}^2)$, this basis is defined by the concatenation of the orthonormal bases of, respectively, $\mathbb{H}^{3},\mathbb{H}^{1,\nabla str},\mathbb{H}^{2}$ and $\mathbb{H}^{1,\nabla rot}$. From this basis let define the three following sets:
\ben
\{H_{i}\}=\{h_{i}\otimes h_{i}\},\quad \{K_{i}\}=\{h_{i+7}\otimes h_{i+15}\},\quad \{K^{T}_{i}\}=\{h_{i+15}\otimes h_{i+7}\}
\een
From the first set we can define the following projectors:
\ben
\mathrm{P}^{\mathbb{H}^3}:=\sum_{i=1}^{7}H_{i},\quad \mathrm{P}^{\mathbb{H}^{1,\nabla str}}:=\sum_{i=8}^{10}H_{i},\quad \mathrm{P}^{\mathbb{H}^2}:=\sum_{i=11}^{15}H_{i},\quad \mathrm{P}^{\mathbb{H}^{1,\nabla rot}}:=\sum_{i=16}^{18}H_{i}
\een
These projectors satisfy the following property:
\ben
\mathrm{Id}_{\mathbb{T}^{2}}=\mathrm{P}^{\mathbb{H}^3}+\mathrm{P}^{\mathbb{H}^{1,\nabla str}}+\mathrm{P}^{\mathbb{H}^2}+\mathrm{P}^{\mathbb{H}^{1,\nabla rot}}
\een
To form a basis for isotropic elements of $\mathbb{C}^{(2)}$, the projectors have to be completed by the two following coupling operators\footnote{In this notation $\mathrm{Q}^{A/B}$ denotes the coupling operator from the  the source space $B$ to the destination space $A$.}:
\ben
\mathrm{Q}^{\nabla str/\nabla rot}:=\sum_{i=1}^{3}K_{i},\quad \mathrm{Q}^{\nabla rot/\nabla str}:=\sum_{i=1}^{3}K^{T}_{i}
\een
Hence, any second order elasticity tensor can be rewritten in terms of these operators:
\ben
\mathrm{A}=m_{s_{3}}\mathrm{P}^{\mathbb{H}^3}+m_{s_{1}}\mathrm{P}^{\mathbb{H}^{1,\nabla str}}+m_{c_{1}}(\mathrm{Q}^{\nabla str/\nabla rot}+\mathrm{Q}^{\nabla rot/\nabla str})+m_{r_{2}}\mathrm{P}^{\mathbb{H}^2}+m_{r_{1}}\mathrm{P}^{\mathbb{H}^{1,\nabla rot}}
\een
This relation is analogous to the relation \eqref{eq:Ela_KG} in classical elasticity.
Now, we can compute the multiplication table of the basis $(\mathrm{P}^{\mathbb{H}^3},\mathrm{P}^{\mathbb{H}^2},\mathrm{P}^{\mathbb{H}^{1,\nabla str}},\mathrm{Q}^{\nabla str/\nabla rot},\mathrm{Q}^{\nabla rot/\nabla str},\mathrm{P}^{\mathbb{H}^{1,\nabla rot}})$:
\ben
   \begin{array}{|c||c|c|c|c|c|c|}
       \hline
       \times&\mathrm{P}^{\mathbb{H}^3}&\mathrm{P}^{\mathbb{H}^2}&\mathrm{P}^{\mathbb{H}^{1,\nabla str}}&\mathrm{Q}^{\nabla str/\nabla rot}
       &\mathrm{Q}^{\nabla rot/\nabla str}&\mathrm{P}^{\mathbb{H}^{1,\nabla rot}}\\ \hline \hline
       \mathrm{P}^{\mathbb{H}^3}&\mathrm{P}^{\mathbb{H}^3}&0&0&0&0&0\\\hline
       \mathrm{P}^{\mathbb{H}^2}&0&\mathrm{P}^{\mathbb{H}^2}&0&0&0&0\\\hline
       \mathrm{P}^{\mathbb{H}^{1,\nabla str}}&0&0&\mathrm{P}^{\mathbb{H}^{1,\nabla str}}&\mathrm{Q}^{\nabla str/\nabla rot}&0&0\\\hline
       \mathrm{Q}^{\nabla str/\nabla rot}&0&0&0&0&\mathrm{P}^{\mathbb{H}^{1,\nabla str}}&\mathrm{Q}^{\nabla str/\nabla rot}\\\hline
       \mathrm{Q}^{\nabla rot/\nabla str}&0&0&\mathrm{P}^{\mathbb{H}^{1,\nabla rot}}&\mathrm{Q}^{\nabla rot/\nabla str}&0&0\\\hline
       \mathrm{P}^{\mathbb{H}^{1,\nabla rot}}&0&0&0&0&\mathrm{P}^{\mathbb{H}^{1,\nabla rot}}&\mathrm{Q}^{\nabla rot/\nabla str}\\\hline
   \end{array}
\een
This multiplication table  corresponds to the irreducible algebra of second-order elasticity tensor as identified by Monchiet and Bonnet in \cite{MB11a}. As a consequence the constitutive law \eqref{eq:Ela_Sec} can be broken down into four relations
\ben\label{eq:Ela-Dec2}
\begin{cases}
\mathrm{\tau}^{\mathbb{H}^3}=m_{s_{3}}\eta^{\mathbb{H}^3}\\
\tau^{\mathbb{H}^2}=m_{r_{2}}\eta^{\mathbb{H}^2}\\
\tau^{\mathbb{H}^{1,\nabla str}}=m_{s_{1}}\eta^{\mathbb{H}^{1,\nabla str}}+m_{c_{1}}\eta^{\mathbb{H}^{1,\nabla rot}}\\
\tau^{\mathbb{H}^{1,\nabla rot}}=m_{c_{1}}\eta^{\mathbb{H}^{1,\nabla str}}+m_{r_{1}}\eta^{\mathbb{H}^{1,\nabla rot}}
\end{cases}
\een
Following Walpole the algebra associated with the basis vectors of our representation can be decomposed into 3 irreducible algebras\footnote{The algebra we obtain is similar to the algebra obtained in classical elasticity for transverse isotropy \cite{Wal84}.}:
\begin{itemize}
\item $\mathcal{A}^{3}$, which is associated with the projection on $\HH^{3}$. Since there is an unique $\HH^{3}$ in the decomposition, $\mathcal{A}^{3}$ is unary;
\item $\mathcal{A}^{2}$, which is associated with the projection on $\HH^{2}$. Since there is an unique $\HH^{2}$ in the decomposition, $\mathcal{A}^{2}$ is unary;
\item $\mathcal{A}^{1}$, which is associated with the projection on $\mathbb{H}^{1,\nabla str}\oplus\mathbb{H}^{1,\nabla rot}$. Since there is 2 $\HH^{1}$ in the decomposition, $\mathcal{A}^{1}$ is isomorphic to the matrix algebra of 2$\times$2 matrices.
\end{itemize}
Therefore our construction based on the harmonic decomposition of $\mathbb{T}^{2}$ allows to construct the Walpole decomposition of tensor $\mathrm{C}^{(2)}=\mathrm{A}$. Our result is agreement with the previous results given by Monchiet and Bonnet in \cite{MB11a} on the same topic. However, contrary to the aforementioned paper, our construction is based on a physical interpretation of the harmonic decomposition of the strain gradient tensor. Therefore the moduli $(m_{s_{3}},m_{s_{1}},m_{r_{2}},m_{r_{1}},m_{c_{1}})$ are, in our representation, related to kinematic mechanisms.\\

In the context of the kinetic interpretation of the harmonic decomposition of $\mathbb{T}^{2}$, the moduli have the following expressions:
\begin{center}
\begin{tabular}{|c||c|c|c|}
       \hline
        Mechanism&Type&Expressions\\\hline\hline
        Stretch-gradient& $\HH^{3}$&$m_{s_{3}}=A_{111111}-A_{111221}-4A_{122133}-2A_{122331}$ \\ \cline{2-3} 
        & $\HH^{1}$ &$m_{s_{1}}= A_{111111}+\frac{2}{3}(A_{111221}+4 A_{122133}+2 A_{122331})$ \\ \hline
        \text{Rotation-gradient}& $\HH^{2}$&$m_{r_{2}}=-\frac{A_{111111}}{2}-A_{111221}+2 A_{122133}+4 A_{122331}+\frac{3 A_{221221}}{2}$ \\ \cline{2-3} 
        & $\HH^{1}$ &$m_{r_{1}}=\frac{1}{6} (-3 A_{111111}+2 A_{111221}+20 A_{122133}-8 A_{122331}+9 A_{221221})$ \\ \hline
        \text{Coupling}&$\HH^{1}$&$m_{c_{1}}=\frac{\sqrt{5}}{3}(-2 A_{111221}+4 A_{122133}+2 A_{122331})$ \\ \hline
\end{tabular}
\end{center}
The second-order elasticity tensors will become singular in the following situations:
\ben
\{m_{s_{3}}=0 ;m_{r_{2}}=0; m_{s_{1}}m_{r_{1}}-m_{c_{1}}^2=0\}
\een
The first two conditions are obvious since $m_{s_{3}}$ and $m_{r_{2}}$ are also eigenvalues of $\mathrm{A}$, the third one is equivalent to the vanishing of the determinant $\begin{vmatrix}
m_{s_{1}} & m_{c_{1}} \\
m_{c_{1}} & m_{r_{1}}
\end{vmatrix}$ of the algebra generated by the vector components of the harmonic decomposition of $\mathbb{T}^2$. As already noted if the coupling modulus $m_{c_{1}}$ is set to zero, the Walpole and the Kelvin's decomposition become equivalent.

This representation of the second order elasticity tensor gives an insigth into the physics encoded in the operator, since the physical meaning of the moduli $\{m_{i}\}$ are rather clear:
\begin{itemize}
\item $m_{s_{3}}$ and $m_{s_{1}}$ are moduli related to the stretch-gradient part of $\eta_{(ij)k}$, $m_{s_{3}}$ is linked with the septor part of the stretch-gradient while $m_{s_{1}}$ concerns the vector part. Furthermore $m_{s_{3}}$ is also an eigenvalue of the operator, hence $m_{s_{3}}$ can not vanish without making $\mathrm{A}$ singular. At the opposite, the vector part $m_{s_{1}}$ can be $0$.
\item $m_{r_{2}}$ and $m_{r_{1}}$ are moduli related to the rotation-gradient part of $\eta_{(ij)k}$, $m_{r_{2}}$ is linked with the septor part of the stretch-gradient while $m_{r_{1}}$ concerns the vector part. As before, $m_{r_{2}}$ is also an eigenvalue of the operator, hence $m_{r_{2}}$ can not vanish without making $\mathrm{A}$ singular. At the opposite, the vector part $m_{r_{1}}$ can be $0$.
\item $m_{c_{1}}$ is the coupling modulus between the rotation-gradient and the vector part of the stretch gradient.
\end{itemize}
These quantities are interesting for, at least, two reasons. First the computation of the different ratio of these quantities are necessary to quantify the relative importance of the different mechanisms and, at the end, to justify to neglect some of them.
In this view the following quantity, which measures the relative effect of rotation-gradient vs stretch-gradient can be defined :
\ben
r/s=\frac{\sqrt{m_{r_{2}}^2+m_{r_{1}}^2}}{\sqrt{m_{s_{3}}^2+m_{s_{1}}^2}}
\een
The second and closely related interest is that it makes easy to properly impose kinematic constraints. For example, if the material is only sensitive to the rotation gradient, the moduli associated with the other mechanisms should be $0$, in other terms 
\ben
m_{s_{3}}=m_{s_{1}}=m_{c_{1}}=0
\een

\section{Construction of the proofs}\label{sec:proof}

The proofs of our theorems are rather direct, and make use of the Schur's lemma. Therefore before constructing our proof, this technical lemma has to be introduced.

\subsection{The Schur's lemma}

Let us define the notion of an intertwining operator.
\begin{defn}
Let $(E_{1},\rho_{1})$ and $(E_{2},\rho_{2})$ be two $G$ representations. An linear operator $\mathrm{T}:E_{1}\rightarrow E_{2}$ is said to be an \emph{intertwining operator} between $\rho_{1}$ and $\rho_{2}$ if the following graph commutes for all $g\in G$:
\begin{equation*}\label{Fig:DiagSchur}
\xymatrix{
      E_{1}\ar[d]^{\rho_{1}}\ar[r]^{\mathrm{T}}	&E_{2}\ar[d]^{\rho_{2}} \\
\ E_{1}\ar[r]^{\mathrm{T}}	&E_{2}}
\end{equation*}
i.e. if
\begin{equation*}
 \forall g\in G,\ \rho_{2}(g)\circ\mathrm{T}=\mathrm{T}\circ\rho_{1}(g)
\end{equation*}
\end{defn}

\begin{thm}[Schur's lemma]
 Let $\mathrm{T}$ be an intertwining operator between irreducible $\mathrm{G}$-representations $(E_{1},\rho_{1})$ and $(E_{2},\rho_{2})$.
\begin{itemize}
 \item If $\rho_{1}$ and $\rho_{2}$ are inequivalent, then $\mathrm{T}=0$;
 \item If $E_{1}=E_{2}=E$ and $\rho_{1}=\rho_{2}=\rho$, then $\mathrm{T}=\lambda\mathrm{Id}_{E}$.
\end{itemize}
\end{thm}

Now let consider an harmonic space $\mathbb{H}^{k}$, and $\mathrm{A}_{k}$ an isotropic self-adjoint endomorphism of $\mathbb{H}^{k}$ we have the following lemma:
\begin{lem}\label{lem.diag}
If $\mathrm{A}_{k}$ is an isotropic self-adjoint endomorphism of $\mathbb{H}^{k}$ then $\mathrm{A}_{k}=\lambda\mathrm{Id}_{\mathbb{H}^{k}}$.
\end{lem}

\begin{proof}
To demonstrate the lemma  we need to show that $\mathrm{A}_{k}$ is an intertwining operator between the irreducible the $\SO(3)$-representation $(\mathbb{H}^{k},\rho_{k})$ and itself.
Therefore we need to  show that the following diagram commutes
\begin{equation*}\label{Fig:DiagSchur2}
\xymatrix{
      \mathbb{H}^{k}\ar[d]^{\rho_{k}}\ar[r]^{\mathrm{A}_{k}}	&\mathbb{H}^{k}\ar[d]^{\rho_{k}} \\
\ \mathbb{H}^{k}\ar[r]^{\mathrm{A}_{k}}	&\mathbb{H}^{k}}
\end{equation*}
i.e. if
\begin{equation*}
\forall g\in G, \ \rho_{k}(g)\circ\mathrm{A}_{k}=\mathrm{A}_{k}\circ\rho_{k}(g)
\end{equation*}
The former condition is equivalent to:
\begin{equation*}
 \mathrm{A}_{k}=\rho_{k}(g)\mathrm{A}_{k}\rho^{-1}_{k}(g)
\end{equation*}
This condition is satisfied for any element of the symmetry group of $\mathrm{A}_{k}$. As $\mathrm{A}_{k}$ is supposed to be isotropic the relation is verified for any $g\in\SO(3)$.
Therefore $\mathrm{A}_{k}$ is an intertwining operator between $\SO(3)$-representations. Since $\mathbb{H}^{k}$ is irreducible a direct application of Schur's lemma terminates the proof.
\end{proof}

\subsection{Main proofs}

\begin{thm}[Theorem 1]
If the harmonic decomposition of $\mathbb{T}$ is unique, and if $\mathrm{C}$ is isotropic then its matrix representation $C_{\mathcal{H}}$ is diagonal. In such a case $d^{\star}$ reduces to a permutation and the Walpole and the Kelvin decompositions are equivalent.
\end{thm}

\begin{proof}
Let consider an $\SO(3)$-irreducible decomposition of $\mathbb{T}$, this decomposition is supposed to be constructed from an isomorphism $h$:
\ben
\mathbb{T}\underset{h}{=}\bigoplus_{k=1}^{p}\mathbb{H}^{\alpha_{k}};\quad \alpha_{p}=n 
\een
where $\{\alpha_{k}\}$ is an increasing series that contains $p$ terms. If $\sharp(\mathbb{T})$ denotes the number of different irreducible spaces involved in the decomposition, we have $\sharp(\mathbb{T})=p$.\\

Now we have to compute the Clebsch-Gordan decomposition of $\mathbb{T}\otimes^S\mathbb{T}$.
Since $\mathbb{T}=\bigoplus_{k=1}^{p}\mathbb{H}^{\alpha_{k}}$, the symmetric tensor product $\mathbb{T}\otimes^S\mathbb{T}$ can be decomposed into a direct sum as: 
\bee\label{dec_prod}
\mathbb{T}\otimes^S\mathbb{T}=\left(\bigoplus_{k=1}^{p}\mathbb{H}^{\alpha_{k}}\right)\otimes^{S}\left(\bigoplus_{k=1}^{p}\mathbb{H}^{\alpha_{k}}\right)=\underbrace{\left(
\bigoplus_{k=1}^{p}\mathbb{H}^{\alpha_{k}}\otimes^S \mathbb{H}^{\alpha_{k}}\right)}_{\mathrm{S}} \oplus \underbrace{\left(\bigoplus_{1\leq i<j \leq p}\mathbb{H}^{\alpha_{i}}\otimes \mathbb{H}^{\alpha_{j}}\right)}_{\mathrm{C}}
\eee
This expression shows that any symmetric endomorphism of $\mathbb{T}$ can be decomposed into a direct sum of:
\begin{enumerate}
\item $\mathrm{S}$-term: $p$ self-adjoint endomorphisms between $\mathbb{H}^{\alpha_{k}}$;
\item $\mathrm{C}$-term: $\frac{p(p-1)}{2}$ coupling morphisms between the $p$ irreducible spaces.
\end{enumerate}
Each of these terms can be decomposed into $\SO(3)$-irreducible spaces, this is done using the following Clebsch-Gordan rules:
\begin{itemize}
\item for the classical product: $\mathbb{H}^{p}\otimes\mathbb{H}^{q}=\bigoplus_{i=|p-q|}^{p+q}\mathbb{H}^{i}$
\item for the symmetric product: $\mathbb{H}^{k}\otimes^{S}\mathbb{H}^{k}=\bigoplus_{i=0}^{k}\mathbb{H}^{2i}$
\end{itemize}
An isotropic space $\left[\mathbb{T}\otimes^S\mathbb{T}\right]^{\SO(3)}$ can only contain scalars, i.e. harmonic spaces of order $0$. Looking at the Clebsch-Gordan products, it appears that only same order products can generate scalar terms.  In the relation \eqref{dec_prod} it can be seen that same order products can either be produced by self-products ($\mathrm{S}$-term) or cross-products between same order terms ($\mathrm{C}$-term).
\ben\label{dec_prod2}
[\mathbb{T}\otimes^S\mathbb{T}]^{\SO(3)}=\bigoplus_{k=1}^{q}\mathbb{H}^{0,k}=\underbrace{\left(\bigoplus_{i=1}^{p}\mathbb{H}^{0,i}\right)}_{\mathrm{S}}\oplus\underbrace{\left(\bigoplus_{j=1}^{q-p}\mathbb{H}^{0,j}\right)}_{\mathrm{C}}
\een
Let us consider now the hypothesis that the decomposition of $\mathbb{T}$ is unique. It is known that it is the case iff the order of spaces in the harmonic decomposition of $\mathbb{T}$ are all different. This hypothesis means:
\ben
\forall k\in[1,p-1],\ \alpha_{k+1}>\alpha_{k}
\een
As a consequence in the relation \eqref{dec_prod} there is no cross-product between same order terms and, hence no isotropic coupling between harmonic spaces. Under this hypothesis we therefore have
\bee\label{dec_prod3}
[\mathbb{T}\otimes^S\mathbb{T}]^{\SO(3)}=\underbrace{\left(
\bigoplus_{k=1}^{p}\mathbb{H}^{\alpha_{k}}\otimes^S \mathbb{H}^{\alpha_{k}}\right)}_{\mathrm{S}}=\underbrace{\left(\bigoplus_{i=1}^{p}\mathbb{H}^{0,i}\right)}_{\mathrm{S}}
\eee
Therefore any $\mathrm{C}\in[\mathbb{T}\otimes^S\mathbb{T}]^{\SO(3)}$ is defined by a collection $\{\mathrm{C}_{k}\}$ of $p$ self-adjoint isotropic endomorphisms. A direct application of the lemma \ref{lem.diag}, shows that for all $k$, $\mathrm{C}_{k}=\lambda\mathrm{Id}_{\mathbb{H}^{k}}$. Therefore under the condition of the theorem the matrix of $\mathrm{C}$ expressed in $\mathcal{H}(\mathbb{T})\otimes^S\mathcal{H}(\mathbb{T})$ is diagonal. Since the matrix of $\mathrm{C}$ expressed in $\mathcal{E}(\mathrm{C})\otimes^S\mathcal{E}(\mathrm{C})$ is also diagonal, they may differ by a permutation. Therefore, in such a case, $d^{\star}$ reduces to a permutation, and the Kelvin and Walpole decompositions are equivalent.
\end{proof}

\begin{cor}
If the harmonic decomposition of $\mathbb{T}$ is unique, then $\dim(\left[\mathbb{T}\otimes^S\mathbb{T}\right]^{\SO(3)})=\sharp(\mathbb{T})$.
\end{cor}

\begin{proof}
It is a direct consequence of the previous theorem. Getting back at the relation \eqref{dec_prod3}, it shows that if $\mathbb{T}$ is unique there is not coupling between harmonic spaces. Therefore the isotropic components $\mathbb{H}^{0}$ are generated by self products of the elements of the harmonic decomposition of $\mathbb{T}$. Hence an element $\mathrm{C}\in \left[\mathbb{T}\otimes^S\mathbb{T}\right]^{\SO(3)}$ has as many components as the number of harmonic components in  $\mathbb{T}$.
\end{proof}

\begin{thm}[Theorem 2]
Let $\mathbb{T}$ be a $n$th-order tensor space only endowed with minor symmetries, then the harmonic decomposition of $\mathbb{T}$ is unique or, equivalently:
\ben
\dim(\left[\mathbb{T}\otimes^S\mathbb{T}\right]^{\SO(3)})=\sharp^{\SO(3)}(\mathbb{T})
\een
if $\mathbb{T}=S^n(\RR^{3}),\  n\geq0$ or $\mathbb{T}=\otimes^2(\RR^{3})$.
\end{thm}

\begin{proof}
Let $\mathbb{T}$ be a subset of $G^{n}=\otimes^{n}(\RR^{3})$ the space of $n$th-order tensor. If this subspace is only defined in terms of index symmetries\footnote{This hypothesis excludes, for example, traceless subspace.} we have:
\ben
S^n(\RR^{3})\subseteq\mathbb{T}\subseteq\otimes^{n}(\RR^{3})
\een
With the following important property \cite{JCB78}:
\ben
S^n(\RR^{3})=
\begin{cases}
\bigoplus_{k=0}^{\frac{n}{2}}\HH^{2k}\ \text{if $n$ even}\\
\bigoplus_{k=0}^{\frac{n-1}{2}}\HH^{2k+1}\ \text{if $n$ odd}\\
\end{cases}
\een
Therefore if $\mathbb{T}=S^n(\RR^{3})$ its harmonic decomposition of $\mathbb{T}$ is unique. At the opposite if $\mathbb{T}=\otimes^{n}(\RR^{3})$, the harmonic decomposition of $\mathbb{T}$ is unique provided $n<3$. This fact is easy to verify, lets compute some decompositions using the Clebsch-Gordan product:
\ben
\otimes^{1}(\RR^{3})=\HH^{1}\quad;\quad \otimes^{2}(\RR^{3})=\HH^{0}\oplus\HH^{1}\oplus\HH^{2}\quad;\quad \otimes^{2}(\RR^{3})=2\HH^{0}\oplus3\HH^{1}\oplus2\HH^{2}\oplus\HH^{3}
\een
Therefore the harmonic decomposition of $\otimes^{3}(\RR^{3})$ is not unique, and since for $n>3$, $\otimes^{n}(\RR^{3})=(\otimes^{n-3}(\RR^{3}))\otimes(\otimes^{3}(\RR^{3}))$, the harmonic decomposition of $\otimes^{n>3}(\RR^{3})$ is not unique.
\end{proof}

\section{Conclusion}
In this paper the algebraic structure of the isotropic $n$th-order gradient elasticity has been investigated. Studying the harmonic tensor decomposition of the $n$-th order strain gradient tensor, it has been shown that, and contrary to classical elasticity, higher-order isotropic constitutive tensors always describe coupled mechanisms. As demonstrated, this fact indeed occurs each time the number of isotropic coefficients exceeds the number of elementary spaces in the harmonic decomposition of the $n$-th order strain-gradient tensor. Besides this general result, an explicit construction of this phenomenon has been provided for Mindlin strain-gradient elasticity. This construction is important for at least two reasons:
\begin{itemize}
\item its principle is general and can be applied in many other situations;
\item for strain-gradient elasticity, it provides a physical interpretation to higher-order moduli. Their explicit knowledge is both interesting to analyze the second-order kinematic of a microstructured media and also to impose (if needed) kinematic constrains on the behavior.
\end{itemize}
It has to be noticed that, as the harmonic decomposition of strain-gradient tensors is not unique, many other interpretations of this decomposition can be proposed. This situation has to be compared with the multiple choices one can made for the classical isotropic moduli. The use of these \emph{physically-based} moduli may be a powerful tool to study the degeneracy that occurs in higher-order continua.

\appendix

\section{Matrix representations of strain-gradient elasticity: main results}\label{Sec:MatRep}

Any second order elasticity tensor $\mathrm{A}$ can be represented by a matrix relative to an orthonormal basis $\{\mathbf{e}_{1},\mathbf{e}_{2},\mathbf{e}_{3}\}$. This question has been fully investigated in \cite{ALH13}. Here we sum-up the main results. As expressed by the relation \eqref{eq:Ela_Sec},  the second-order elasticity tensor $\mathrm{A}$ is a self-adjoint endomorphism of $\mathbb{T}_{(ij)k}$. In order to express the strain gradient $\mathbf{\eta}$
\ or the hyperstress tensor $\mathbf{\tau}$ as a 18-dimensional vector and write $\mathrm{A}$\ as a symmetric matrix, we introduce the following orthonormal basis vectors:%
\begin{equation}
\mathbf{\hat{e}}_{\alpha }=\left( \frac{1-\delta _{ij}}{\sqrt{2}}+\frac{%
\delta _{ij}}{2}\right) \left( \mathbf{e}_{i}\otimes \mathbf{e}_{j}+\mathbf{e%
}_{j}\otimes \mathbf{e}_{i}\right) \otimes \mathbf{e}_{k},\quad 1\leq \alpha
\leq 18  \label{Basis}
\end{equation}%
where the Einstein summation convention does not apply. Then, the aforementioned tensors can be expressed as 
\begin{equation}
\mathbf{\eta}=\displaystyle\sum\limits_{\alpha =1}^{18}\hat{\eta}_{\alpha }\mathbf{\hat{e}}_{\alpha },\text{ \ \ \ }
\mathbf{\tau }=\displaystyle\sum\limits_{\alpha =1}^{18}\hat{\tau}_{\alpha }\mathbf{\hat{e}}_{\alpha },\text{ \ \ \ }
\mathrm{A}=\displaystyle\sum\limits_{\alpha ,\beta=1}^{18}\hat{A}_{\alpha \beta }\mathbf{\hat{e}}_{\alpha }\otimes \mathbf{\hat{e}}_{\beta },  \label{wtA}
\end{equation}%
so that the relation \eqref{eq:Ela_Sec} can be conveniently written in the matrix form
\begin{equation}
\hat{\tau}_{\alpha }=\hat{A}_{\alpha \beta }\hat{\eta}_{\beta}.
\label{SGER}
\end{equation}%
Using the orthonormal basis (\ref{Basis}), the relationship between the matrix components $\hat{\eta}_{\alpha }$ and $\eta _{ijk}$
\begin{equation}
\hat{\eta}_{\alpha }=%
\begin{cases}
\eta _{ijk}\text{ \ if \ }i=j, \\ 
\sqrt{2}\eta _{ijk}\text{ \ if \ }i\neq j;%
\end{cases}%
\label{3-to-1}
\end{equation}%
and, obviously the same relation  between $\hat{\tau}_{\alpha }$ and $\tau _{ijk}$ hold. Fot the constitutive tensor we have the following correspondence.
\begin{equation}
\hat{A}_{\alpha \beta }=%
\begin{cases}
A_{ijklmn}\text{ \ if \ }i=j\ \text{and}\ l=m\text{,} \\ 
\sqrt{2}A_{ijklmn}\text{ \ if \ }i\neq j\ \text{and}\ l=m\ \text{or}\ i=j\ 
\text{and}\ l\neq m, \\ 
2A_{ijklmn}\text{ \ if \ }i\neq j\ \text{and}\ l\neq m.%
\end{cases}
\label{6-to-3}
\end{equation}
It remains to choose an appropriate three-to-one subscript correspondence between $ijk$ and $\alpha $. The correspondence specified in Table 1 was chosen in order to make the 6th order tensor block diagonal for diehdral classes \cite{ALH13}.
\begin{table}[tbp]
\begin{center}
\begin{tabular}{|c||c|c|c|c|c||c|}
\hline
$\alpha$ & $1$ & $2$ & $3$&$ 4 $&$ 5 $&  \\ \hline\hline
$ijk $&$ 111 $&$ 221 $&$ 122 $&$ 331 $&$ 133$ & Privileged direction: $1$ \\ 
\hline\hline
$\alpha $&$ 6 $&$ 7 $&$ 8 $&$ 9 $&$ 10$ &  \\ \hline
$ijk $&$ 222 $&$ 112 $&$ 121 $&$ 332 $&$ 233$ & Privileged direction: $2$ \\ 
\hline\hline
$\alpha $&$ 11 $&$ 12 $&$ 13 $&$ 14 $&$ 15$ &  \\ \hline
$ijk $&$ 333 $&$ 113 $&$ 131 $&$ 223 $&$ 232$ & Privileged direction: $3$ \\ 
\hline\hline
$\alpha $&$ 16 $&$ 17 $&$ 18$ &  &  &  \\ \hline
$ijk $&$ 123 $&$ 132 $&$ 231$ &  &  & No privileged direction \\ \hline
\end{tabular}
\end{center}
\caption{The three-to-one subscript correspondence for 3D strain-gradient elasticity}
\end{table}
This constitute the so-called spatial basis $\mathcal{S}$ discussed in the main part of the paper. The knowledge ogf the harmonic decomposition of $\mathbb{T}_{(ij)k}$, studied in section \ref{ss:DecT2}, allows to construct a transformation matrix from $\mathcal{S}$ to $\mathcal{H}(\mathbb{T}_{(ij)k})$. This matrix is decomposed here into four elementary matrices, each of them associated with an harmonic space of $\mathbb{T}_{(ij)k}$.
\ben
[\mathrm{P}_{\mathbb{H}^3}]=
\begin{pmatrix}
 \sqrt{\frac{2}{5}} & 0 & 0 & 0 & 0 & 0 & 0\\
 -\frac{1}{\sqrt{10}} & \frac{1}{\sqrt{6}} & 0 & 0 & 0 & 0 &0\\
 -\frac{1}{\sqrt{5}} & \frac{1}{\sqrt{3}} & 0 & 0 & 0 & 0 & 0 \\
 -\frac{1}{\sqrt{10}} & -\frac{1}{\sqrt{6}} & 0 & 0 & 0 & 0 & 0 \\
 -\frac{1}{\sqrt{5}} & -\frac{1}{\sqrt{3}} & 0 & 0 & 0 & 0 & 0 \\
 0 & 0 & \sqrt{\frac{2}{5}} & 0 & 0 & 0 & 0  \\
 0 & 0 & -\frac{1}{\sqrt{10}} & \frac{1}{\sqrt{6}} & 0 & 0 & 0\\
 0 & 0 & -\frac{1}{\sqrt{5}} & \frac{1}{\sqrt{3}} & 0 & 0 & 0 \\
 0 & 0 & -\frac{1}{\sqrt{10}} & -\frac{1}{\sqrt{6}} & 0 & 0 & 0 \\
 0 & 0 & -\frac{1}{\sqrt{5}} & -\frac{1}{\sqrt{3}} & 0 & 0 & 0 \\
 0 & 0 & 0 & 0 & \sqrt{\frac{2}{5}} & 0  \\
 0 & 0 & 0 & 0 & -\frac{1}{\sqrt{10}} & \frac{1}{\sqrt{6}} & 0 \\
 0 & 0 & 0 & 0 & -\frac{1}{\sqrt{5}} & \frac{1}{\sqrt{3}} & 0 \\
 0 & 0 & 0 & 0 & -\frac{1}{\sqrt{10}} & -\frac{1}{\sqrt{6}} & 0 \\
 0 & 0 & 0 & 0 & -\frac{1}{\sqrt{5}} & -\frac{1}{\sqrt{3}} & 0 \\
 0 & 0 & 0 & 0 & 0 & 0 & \frac{1}{\sqrt{3}} \\
 0 & 0 & 0 & 0 & 0 & 0 & \frac{1}{\sqrt{3}} \\
 0 & 0 & 0 & 0 & 0 & 0 & \frac{1}{\sqrt{3}} 
 \end{pmatrix}\ ,\ 
[\mathrm{P}_{\mathbb{H}^{1,\nabla str}}]=
\begin{pmatrix}
  \sqrt{\frac{3}{5}} & 0 & 0 \\
 \frac{1}{\sqrt{15}} & 0 & 0 \\
 \sqrt{\frac{2}{15}} & 0 & 0 \\
 \frac{1}{\sqrt{15}} & 0 & 0 \\
 \sqrt{\frac{2}{15}} & 0 & 0 \\
 0 & \sqrt{\frac{3}{5}} & 0 \\
 0 & \frac{1}{\sqrt{15}} & 0 \\
 0 & \sqrt{\frac{2}{15}} & 0 \\
 0 & \frac{1}{\sqrt{15}} & 0 \\
 0 & \sqrt{\frac{2}{15}} & 0 \\
 0 & 0 & \sqrt{\frac{3}{5}} \\
 0 & 0 & \frac{1}{\sqrt{15}} \\
 0 & 0 & \sqrt{\frac{2}{15}} \\
 0 & 0 & \frac{1}{\sqrt{15}} \\
 0 & 0 & \sqrt{\frac{2}{15}} \\
 0 & 0 & 0 \\
 0 & 0 & 0 \\
 0 & 0 & 0
 \end{pmatrix}
\een

\ben
[\mathrm{P}_{\mathbb{H}^2}]=
\begin{pmatrix}
   0 & 0 & 0 & 0 & 0 \\
 -\frac{1}{\sqrt{3}} & 0 & 0 & 0 & 0 \\
 \frac{1}{\sqrt{6}} & 0 & 0 & 0 & 0 \\
 \frac{1}{\sqrt{3}} & 0 & 0 & 0 & 0 \\
 -\frac{1}{\sqrt{6}} & 0 & 0 & 0 & 0 \\
 0 & 0 & 0 & 0 & 0 \\
 0 & \frac{1}{\sqrt{3}} & 0 & 0 & 0 \\
 0 & -\frac{1}{\sqrt{6}} & 0 & 0 & 0 \\
 0 & -\frac{1}{\sqrt{3}} & 0 & 0 & 0 \\
 0 & \frac{1}{\sqrt{6}} & 0 & 0 & 0 \\
 0 & 0 & 0 & 0 & 0 \\
 0 & 0 & -\frac{1}{\sqrt{3}} & 0 & 0 \\
 0 & 0 & \frac{1}{\sqrt{6}} & 0 & 0 \\
 0 & 0 & \frac{1}{\sqrt{3}} & 0 & 0 \\
 0 & 0 & -\frac{1}{\sqrt{6}} & 0 & 0 \\
 0 & 0 & 0 & \frac{1}{\sqrt{2}} & -\frac{1}{\sqrt{6}} \\
 0 & 0 & 0 & 0 & \sqrt{\frac{2}{3}} \\
 0 & 0 & 0 & -\frac{1}{\sqrt{2}} & -\frac{1}{\sqrt{6}}
 \end{pmatrix},\ [\mathrm{P}_{\mathbb{H}^{1,\nabla rot}}]=
\begin{pmatrix}
 0 & 0 & 0 \\
 -\frac{1}{\sqrt{3}} & 0 & 0 \\
 \frac{1}{\sqrt{6}} & 0 & 0 \\
 -\frac{1}{\sqrt{3}} & 0 & 0 \\
 \frac{1}{\sqrt{6}} & 0 & 0 \\
 0 & 0 & 0 \\
 0 & -\frac{1}{\sqrt{3}} & 0 \\
 0 & \frac{1}{\sqrt{6}} & 0 \\
 0 & -\frac{1}{\sqrt{3}} & 0 \\
 0 & \frac{1}{\sqrt{6}} & 0 \\
 0 & 0 & 0 \\
 0 & 0 & -\frac{1}{\sqrt{3}} \\
 0 & 0 & \frac{1}{\sqrt{6}} \\
 0 & 0 & -\frac{1}{\sqrt{3}} \\
 0 & 0 & \frac{1}{\sqrt{6}} \\
 0 & 0 & 0 \\
 0 & 0 & 0 \\
 0 & 0 & 0
  \end{pmatrix}
\een 

\bibliographystyle{unsrt}

\end{document}